\documentclass[a4paper,USenglish,cleveref, autoref, thm-restate]{lipics-v2019}
\nolinenumbers

\title{Synchronization of Deterministic Visibly Push-Down Automata}

\author{Henning Fernau}{Universität Trier, Fachbereich IV, Informatikwissenschaften, Germany} {fernau@uni-trier.de}{https://orcid.org/0000-0002-4444-3220}{}

\author{Petra Wolf}{Universität Trier, Fachbereich IV, Informatikwissenschaften, Germany \and \url{https://www.wolfp.net/}}{wolfp@informatik.uni-trier.de}{https://orcid.org/0000-0003-3097-3906}{DFG project  
	FE 560/9-1}

\authorrunning{H.\ Fernau, P.\ Wolf} 

\Copyright{Henning Fernau, Petra Wolf} 



\category{} 

\relatedversion{A full version of the paper is available at ArXiv~\cite{vpda-crossref}.} 

\supplement{}

\EventEditors{John Q. Open and Joan R. Access}
\EventNoEds{2}
\EventLongTitle{42nd Conference on Very Important Topics (CVIT 2016)}
\EventShortTitle{CVIT 2016}
\EventAcronym{CVIT}
\EventYear{2016}
\EventDate{December 24--27, 2016}
\EventLocation{Little Whinging, United Kingdom}
\EventLogo{}
\SeriesVolume{42}
\ArticleNo{23}

\usepackage{float}
\usepackage{todonotes}
\usepackage{amsmath}

\newcommand{\NP}{\textsf{NP}}
\newcommand{\PSPACE}{\textsf{PSPACE}}
\newcommand{\NPSPACE}{\textsf{NPSPACE}}
\newcommand{\NLOGSPACE}{\textsf{NLOGSPACE}}
\newcommand{\PTIME}{\textsf{P}}

\newcommand{\EXP}{\textsf{EXPTIME}}

\newcommand{\mto}[1][w]{\ensuremath{\overset{#1}{\longrightarrow}}}
\DeclareMathOperator{\lang}{\mathcal{L}}
\newcommand{\cerny}{{\v{C}}ern{\'{y}}}
\usepackage{soul}
\usepackage{wasysym}
\usepackage{booktabs}
\usepackage{tikz}
\usetikzlibrary{arrows,automata}
\usepackage{apxproof}

\begin{document}

\maketitle

\begin{abstract}
We generalize the concept of synchronizing words for finite automata, which map all states of the automata to the same state, to deterministic visibly push-down automata. Here, a synchronizing word $w$ does not only map all states to the same state but also fulfills some conditions on the stack content of each run after reading $w$. We consider three types of these stack constraints: after reading~$w$, the stack (1) is empty in each run, (2) contains the same sequence of stack symbols in each run, or (3) contains an arbitrary sequence which is independent of the other runs. We show that in contrast to general deterministic push-down automata, it is decidable for deterministic visibly push-down automata whether there exists a synchronizing word with each of these stack constraints, more precisely, the problems are in \EXP. Under the constraint (1), the problem is even in \PTIME. For the sub-classes of deterministic very visibly push-down automata, the problem is in \PTIME\ for all three types of  constraints. We further study variants of the synchronization problem where the number of turns in the stack height behavior caused by a synchronizing word is restricted, as well as the problem of synchronizing a variant of a sequential transducer, which shows some visibly behavior, by a word that synchronizes the states and produces the same output on all runs.
\end{abstract}

\section{Introduction}
The classical \emph{synchronization problem} asks, given a deterministic finite automaton (DFA), whether there exists a \emph{synchronizing word} that brings all states of the automaton to a single state. While this problem is solvable in polynomial time~\cite{cerny1964,San2005,Vol2008}, many variants, such as synchronizing only a subset of states~\cite{San2005}, or synchronizing a partial automaton without taking an undefined transition (called carefully synchronizing)~\cite{DBLP:journals/mst/Martyugin14}, are \PSPACE-complete. Restricting the length of a potential synchronizing word by a parameter in the input also yields a harder problem, namely the \NP-complete short synchronizing word problem~\cite{Rys80,DBLP:journals/siamcomp/Eppstein90}.
The field of synchronizing automata has been intensively studied over the last years, among others in attempt to verify the famous \cerny\ conjecture claiming that every synchronizable DFA admits a synchronizing word of quadratic length in the number of states~\cite{cerny1964,DBLP:journals/jalc/Cerny19,DBLP:journals/eik/Starke66b,DBLP:journals/jalc/Starke19}. The currently best upper bound on this length is cubic, and only very little progress has been made, basically improving on the multiplicative constant factor in front of the cubic term, see~\cite{DBLP:journals/jalc/Shitov19, DBLP:conf/stacs/Szykula18}.
More information on synchronization of DFA and the \cerny\ conjecture can be found in~\cite{Vol2008,beal_perrin_2016,JALC20}.
In this work, we want to move away from deterministic finite automata to more general deterministic visibly push-down automata.\footnote{
The term \emph{synchronization of push-down automata}  already  occurs in the literature, i.e., in ~\cite{caucal2006synchronization,DBLP:journals/mst/ArenasBL11}, but there the term \emph{synchronization} refers to some relation of the input symbols to the stack behavior~\cite{caucal2006synchronization} or to reading different words in parallel~\cite{DBLP:journals/mst/ArenasBL11}; do not to confuse it with our notion of synchronizing states.}

The synchronization problem has been generalized in the literature to other automata models including infinite-state systems with infinite branching such as weighted and timed automata~\cite{DBLP:conf/fsttcs/0001JLMS14,DBLP:phd/hal/Shirmohammadi14} or register automata~\cite{babari2016synchronizing}. 
For instance, register automata are infinite state systems where a state consists of a control state and register contents. 

Another automaton model, where the state set is enhanced with a potential infinite memory structure, namely a stack, is the class of \emph{nested word automata} (NWAs were introduced in~\cite{DBLP:journals/jacm/AlurM09}), where an input word is enhanced with a matching relation determining at which pair of positions in a word a symbol is pushed to and popped from the stack. The class of languages accepted by NWAs is identical to the class of \emph{visibly push-down languages} (VPL) accepted by \emph{visibly push-down automata} (VPDA) and form a proper sub-class of the deterministic context-free languages. 
VPDAs have first been studied by Mehlhorn~\cite{DBLP:conf/icalp/Mehlhorn80} under the name \emph{input-driven pushdown automata} and became quite popular more recently due to the work by Alur and Madhusudan~\cite{DBLP:conf/stoc/AlurM04}, showing that VPLs share several nice properties with regular languages. For more on VPLs we refer to the survey~\cite{DBLP:journals/sigact/OkhotinS14}.
In~\cite{DBLP:journals/jalc/ChistikovMS19}, the synchronization problem for NWAs was studied. There, the concept of synchronization was generalized to bringing all states to one single state such that for all runs the stack is empty (or in its start configuration) after reading the synchronizing word. In this setting, the synchronization problem is solvable in polynomial time (again indicating similarities of VPLs with regular languages), while the short synchronizing word problem (with length bound given in binary) is \PSPACE-complete; the question of synchronizing from or into a subset is \EXP-complete. Also, matching exponential upper bounds on the length of a synchronizing word are given.

Our attempt in this work is to study the synchronization problem for real-time (no $\epsilon$-transitions) deterministic visibly push-down automata (DVPDA) and several sub-classes thereof, like real-time deterministic very visibly push-down automata (DVVPDA for short; this model was introduced in~\cite{DBLP:phd/dnb/Ludwig19}), real-time deterministic visibly counter automata (DVCA for short; this model appeared a.o.~in \cite{DBLP:conf/stacs/BaranyLS06,DBLP:journals/corr/abs-0901-2068,DBLP:conf/fsttcs/Bollig16,DBLP:conf/mfcs/HahnKLL15,DBLP:conf/dlt/KrebsLL15,DBLP:conf/stacs/KrebsLL15}) and finite turn variants thereof.
We want to point out that, despite the equivalence of the accepted language class, the automata models of nested word automata and visibly push-down automata still differ and the results from~\cite{DBLP:journals/jalc/ChistikovMS19} do not immediately transfer to VPDAs.
In general, the complexity of the synchronization problem can differ for different automata models accepting the same language class. For instance, in contrast to the polynomial time solvable synchronization problem for DFAs, the generalized synchronization problem for finite automata with one ambiguous transition is \PSPACE-complete, as well as the problem of carefully synchronizing a DFA with one undefined transition~\cite{DBLP:conf/wia/Martyugin12}.
We will not only consider the synchronization model introduced in~\cite{DBLP:journals/jalc/ChistikovMS19}, where reading a synchronizing word results in an empty stack on all runs; but we will also consider a synchronization model where not only the final state on every run must be the same but also the stack content needs to be identical, as well as a model where only the states needs to by synchronized and the stack content might be arbitrary. 
These three models of synchronization have been introduced in~\cite{MY20}, where length bounds on a synchronizing word for general DPDAs have been studied dependent on the stack height.
The complexity of these three concepts of synchronization for general DPDAs are considered in~\cite{dpda-crossref} where it is shown that synchronizability is undecidable for general DPDAs and deterministic counter automata (DCA). It becomes decidable for deterministic  partially blind counter automata and is \PSPACE-complete for some types of finite turn DPDAs, while it is still undecidable for other types of finite turn DPDAs.

In contrast, we will show in the following that for DVPDAs and considered sub-classes hereof, the synchronization problem for all three stack models, with restricted or unrestricted number of turns, is in \EXP\ and hence decidable. 
For DVVPDAs and DVCAs, the synchronization problems for all three stack models (with unbounded number of turns) are even in \PTIME.
Like the synchronization problem for NWAs in the empty stack model considered in~\cite{DBLP:journals/jalc/ChistikovMS19}, we observe that the synchronization problem for DVPDAs in the empty stack model is solvable in polynomial time, whereas synchronization of DVPDAs in the same and arbitrary stack models is at least \PSPACE-hard.
If the number of turns caused by a synchronizing word on each run is restricted, the synchronization problem becomes \PSPACE-hard for all considered automata models for $n>0$ and is only in \PTIME\ for $n=0$ in the empty stack model.
We will further introduce variants of synchronization problems distinguishing the same and arbitrary stack models by showing complementary complexities in these models. For problems considered in~\cite{dpda-crossref}, these two stack models have always shared their complexity status.



Missing proof details can be found in the appendix.

\section{Fixing Notations}
\label{sec:Notation}
We refer to the empty word as $\epsilon$.
For a finite alphabet $\Sigma$ we denote with $\Sigma^*$ the set of all words over $\Sigma$ and with $\Sigma^+ = \Sigma\Sigma^*$ the set of all non-empty words.
For $i\in \mathbb{N}$ we set $[i] = \{1, 2, \ldots, i\}$.
For $w \in \Sigma^*$ we denote with $|w|$ the length of $w$, with $w[i]$ for $i \in [|w|]$ the $i$'th symbol of $w$, and with $w[i..j]$ for $i, j \in [|w|]$ the subword $w[i]w[i+1] \ldots w[j]$ of $w$. We call $w[1..i]$ a prefix and $w[i..|w|]$ a suffix of $w$. 
If $i<j$, then $w[j,i]=\epsilon$.

We call $A= (Q, \Sigma, \delta, q_0, F)$ a \emph{deterministic finite automaton} (DFA for short) if $Q$ is a finite set of states, $\Sigma$ is a finite input alphabet, $\delta$ is a transition function $Q \times \Sigma \to Q$, $q_0$ is the initial state and $F \subseteq Q$ is the set of final states. The transition function $\delta$ is generalized to words by $\delta(q, w) = \delta(\delta(q, w[1]), w[2..|w|])$ for $w \in \Sigma^*$.
A word $w\in \Sigma^*$ is accepted by $A$ if $\delta(q_0, w) \in F$ and the language accepted by $A$ is defined by $\lang(A) = \{w \in \Sigma^* \mid \delta(q_0, w) \in F\}$.
We extend $\delta$ to sets of states  $Q'\subseteq Q$  or to sets of letters $\Sigma'\subseteq \Sigma$, letting  $\delta(Q',\Sigma')=\{\delta(q',\sigma')\mid  (q',\sigma')\in Q'\times\Sigma'\}$. Similarly, we may  write $\delta(Q',\Sigma')=p$ to define $\delta(q',\sigma')=p$ for each $(q',\sigma')\in Q'\times\Sigma'$. 
The synchronization problem for DFAs (called \textsc{DFA-Sync}) asks for a given DFA $A$ whether there exists a synchronizing word for $A$. A word $w$ is called a \emph{synchronizing word} for a DFA~$A$, if it brings all states of the automaton to one single state, i.e., $|\delta(Q, w)| = 1$.

We call $M = (Q, \Sigma, \Gamma, \delta, q_0, \bot, F)$ a \emph{deterministic push-down automaton} (DPDA for short) if $Q$ is a finite set of states; the finite sets $\Sigma$ and $\Gamma$ are the input and stack alphabet, respectively; $\delta$ is a transition function $Q \times \Sigma \times \Gamma \to Q \times \Gamma^*$; 
 $q_0$ is the initial state; $\bot \in \Gamma$ is the stack bottom symbol which is only allowed as the first (lowest) symbol in the stack, i.e., if $\delta(q,a,\gamma)=(q',\gamma')$ and $\gamma'$ contains $\bot$, then $\bot$ only occurs in $\gamma'$ as its prefix and moreover, $\gamma=\bot$; and $F$ is the set of final states. 
We will only consider \emph{real-time} push-down automata and forbid $\epsilon$-transitions, as can be seen in the definition.
Notice that the bottom symbol can be removed, but then the computation gets stuck.

Following~\cite{DBLP:journals/jalc/ChistikovMS19}, a \emph{configuration} of $M$ is a tuple $(q, \upsilon) \in Q \times \Gamma^*$.
 For a letter $\sigma \in \Sigma$ and a stack content $\upsilon$ with $|\upsilon| = n$, we write $(q, \upsilon) \mto[\sigma] (q', \upsilon[1..(n-1)]\gamma)$ if $\delta(q, \sigma, \upsilon[n]) = (q', \gamma)$.
This means that the top of the stack $\upsilon$ is the right end of  $\upsilon$.  
We also denote with $\longrightarrow$ the reflexive transitive closure of the union of $\mto[\sigma]$ over all letters in $\Sigma$. The input words on top of $\longrightarrow$ are concatenated accordingly, so that $\longrightarrow\, =\bigcup_{w\in\Sigma^*}\mto[w]$.
The language $\lang(M)$ accepted by a DPDA $M$ is 
$\lang(M) = \{w \in \Sigma^* \mid (q_0, \bot) \mto[w] (q_f, \gamma), q_f \in F\}$.
We call the sequence of configurations $(q, \bot) \mto[w] (q', \gamma)$ the \emph{run} induced by $w$, starting in $q$, and ending in $q'$. We might also call $q'$ the \emph{final state} of the run. 

We will discuss three different concepts of synchronizing DPDAs. For all concepts we demand that a synchronizing word $w \in \Sigma^*$ maps all states, starting with an empty stack, to the same synchronizing state, i.e., for all
$q, q' \in Q \colon (q, \bot) \mto (\overline{q}, \upsilon), (q', \bot) \mto (\overline{q}, \upsilon')$. In other words, for a synchronizing word all runs started on some states in $Q$ end up in the same state.
 In addition to synchronizing the states of a DPDA we will consider the following two conditions for the stack content: 
 (1) $\upsilon = \upsilon' = \bot$, 
 (2) $\upsilon = \upsilon'$. 
We will call (1) the \emph{empty stack model} and (2) the \emph{same stack model}. In the third case, we do not put any restrictions on the stack content and call this the \emph{arbitrary stack model}.

\noindent
As we are only interested in synchronizing a DPDA we can neglect the start and final states.

Starting from DPDAs we define the following sub-classes thereof:
\begin{itemize} 
%
\item A \emph{deterministic visibly push-down automaton} (DVPDA) is a DPDA where the input alphabet~$\Sigma$ can be partitioned into $\Sigma = \Sigma_{\text{call}} \cup \Sigma_{\text{int}} \cup \Sigma_{\text{ret}}$ such that the change in the stack height is determined by the partition of the alphabet. To be more precise, the transition function $\delta$ is modified such that it can be partitioned accordingly into $\delta = \delta_{\text{c}} \cup \delta_{\text{i}} \cup \delta_{\text{r}}$ such that $\delta_{\text{c}} \colon Q \times \Sigma \to Q \times (\Gamma\backslash\{\bot\})$  puts a symbol on the stack, $\delta_{\text{i}} \colon Q \times \Sigma \to Q$ leaves the stack unchanged, and $\delta_{\text{r}} \colon Q \times \Sigma \times \Gamma \to Q$ reads and pops a symbol from the stack~\cite{DBLP:conf/stoc/AlurM04}. If~$\bot$ is the symbol on top of the stack, then $\bot$ is only read and not popped. We call letters in $\Sigma_{\text{call}}$ \emph{call} or \emph{push} letters; letter in $\Sigma_{\text{int}}$ \emph{internal} letters; and letters in $\Sigma_{\text{ret}}$ \emph{return} or \emph{pop} letters.
The language class accepted by DVPDA is equivalent to the class of languages accepted by deterministic nested word automata (see~\cite{DBLP:journals/jalc/ChistikovMS19}).
\item A \emph{deterministic very visibly push-down automaton} (DVVPA) is a DVPDA where not only the stack height but also the stack content is completely determined by the input alphabet, i.e., for a letter $\sigma \in \Sigma$ and all states $p, q \in Q$ for $\delta_{\text{c}}(p, \sigma) = (p', \gamma_p)$ and $\delta_{\text{c}}(q, \sigma) = (q', \gamma_q)$ it holds that $\gamma_p = \gamma_q$.
\item A \emph{deterministic visibly (one) counter automaton} (DVCA) is a DVPDA where $|\Gamma\backslash\{\bot\}| = 1$; note that every DVCA is also a DVVPDA.
\end{itemize}
We are now ready to define a family of synchronization problems, the complexity of which will be our subject of study in the following chapters.
\begin{definition}[\sc Sync-DVPDA-Empty]
	\ \\
	Given: DPDA $M = (Q, \Sigma, \Gamma, \delta, \bot)$.\\
	Question: Does there exist a word $w \in \Sigma^*$ that synchronizes $M$ in the empty stack model?
\end{definition}
For the same stack model, we refer to the synchronization problem above as \textsc{Sync-DVPDA-Same} and as \textsc{Sync-DVPDA-Arb} in the arbitrary stack model. Variants of these problems are defined by replacing the DVPDA in the definition above by a DVVPDA, and DVCA.
If results hold for several stack models or automata models, then we summarize the problems by using set notations in the corresponding statements.
For the problems \textsc{Sync-DVPDA-Same} and \textsc{Sync-DVPDA-Arb} we introduce two further refined variants of these problems, denoted by the extension \textsc{-Return} and \textsc{-NoReturn}, where for all input DVPDA in the former variant  $\Sigma_{\text{ret}}\neq \emptyset$ holds, whereas in the latter variant $\Sigma_{\text{ret}} = \emptyset$ holds. In the following these variants reveal insights in the differences between synchronization in the same stack and arbitrary stack models, as well as connections to a concept of trace-synchronizing a sequential transducer showing some visibly behavior.


We will further consider synchronization of these automata classes in a finite-turn setting.
%
Finite-turn push-down automata are introduced in~\cite{ginsburg1966finite}.
We adopt the definition in~\cite{valiant1973decision}.
For a DVPDA $M$ an \emph{upstroke} of $M$ is a sequence of configurations induced by an input word~$w$ such that no transition decreases the stack-height. Accordingly, a \emph{downstroke} of $M$ is a sequence of configurations in which no transition increases the stack-height. A stroke is either an upstroke or downstroke. 
A DVPDA $M$ is an $n$-turn DVPDA if for all $w \in \lang(M)$ the sequence of configurations induced by $w$ can be split into at most $n+1$ strokes. Especially, for 1-turn DVPDAs each sequence of configurations induced by an accepting word consists of one upstroke followed by a most one downstroke.
Two subtleties arise when translating this concept to synchronization: (a) there is no initial state so that there is no way to associate a stroke counter to a state, and (b) there is no language of accepted words that restricts the set of words on which
the number of strokes should be limited. 
We therefore generalize the concept of finite-turn DVPDAs to finite-turn synchronization for DVPDAs as follows.
\begin{definition} \textsc{$n$-Turn-Sync-DVPDA-Empty}
	\ \\
	Given: DVPDA $M = (Q, \Sigma, \Gamma, \delta, q_0, \bot, F)$.\\
	Question: Is there a synchronizing word $w\in \Sigma^*$ in the empty stack model, such that 
	for all states $q \in Q$, the sequence of configurations $(q, \bot) \mto[w] (\overline{q}, \bot)$ consists of at most $n+1$ strokes?
\end{definition}
We call such a synchronizing word $w$ an \emph{$n$-turn synchronizing word} for $M$.
We define \textsc{$n$-Turn-Sync-DVPDA-Same} and \textsc{$n$-Turn-Sync-DVPDA-Arb} accordingly for the same stack and arbitrary stack model. Further we extend the problem definition to other classes of automata such as real-time DVVPDAs, and DVCAs.
Table~\ref{tab:VPA-results} summarizes our results, obtained in the next sections, on the complexity status of these problems together with the above introduced synchronization problems.
\begin{table}
	\begin{tabular}{lccc}
		class of automata & empty stack model & same stack model & arbitrary stack model\\
		\toprule
		DVPDA & \PTIME & \PSPACE-hard & \PSPACE-hard\\
		DVPDA-NoReturn & \PTIME & \PSPACE-complete & \PTIME\\
		DVPDA-Return & \PTIME & \PTIME & \PSPACE-hard\\
		$n$-Turn-Sync-DVPDA & \PSPACE-hard & \PSPACE-hard & \PSPACE-hard\\
		0-Turn-Sync-DVPDA & \PTIME & \PSPACE-complete & \PSPACE-complete\\
		DVVPDA & \PTIME & \PTIME & \PTIME\\
		$n$-Turn-Sync-DVVPDA & \PSPACE-hard & \PSPACE-hard & \PSPACE-hard\\
		0-Turn-Sync-DVVPDA & \PTIME & \PSPACE-complete & \PSPACE-complete\\
		DVCA & \PTIME & \PTIME & \PTIME\\
		$n$-Turn-Sync-DVCA & \PSPACE-hard & \PSPACE-hard & \PSPACE-hard\\
		1-Turn-Sync-DVCA & \PSPACE-complete & \PSPACE-complete & \PSPACE-complete\\
		0-Turn-Sync-DVCA & \PTIME & \PSPACE-complete & \PSPACE-complete\\
	\end{tabular}
	\caption{Complexity status of the synchronization problem for different classes of deterministic real-time visibly push-down automata in different stack synchronization modes. For the $n$-turn synchronization variants, $n$ takes all values  not explicitly listed. All our problems are  in \EXP.}
	\label{tab:VPA-results}
\end{table} 

Finally, we introduce two \PSPACE-complete problems for DFAs to reduce from later.
\begin{definition}[\textsc{DFA-Sync-Into-Subset} (\PSPACE-complete~\cite{DBLP:journals/ipl/Rystsov83})]
	\ \\
	Given: DFA $A = (Q, \Sigma, \delta)$, subset $S\subseteq Q$.\\
	Question: Is there a word $w\in \Sigma^*$ such that $\delta(Q, w) \subseteq S$?
\end{definition}
\begin{definition}[\textsc{DFA-Sync-From-Subset} (\PSPACE-complete~\cite{San2005})]
	\ \\
	Given: DFA $A = (Q, \Sigma, \delta)$ with $S \subseteq Q$.\\
	Question: Is there a word $w \in \Sigma^*$ that synchronizes $S$, i.e., for which $|\delta(S, w)| = 1$ is true?
\end{definition}
\section{DVPDAs -- Distinguishing the Stack Models}
\label{sec:DVPDA}
We start with some positive result showing that we come down from the undecidability of the synchronization problem for general DPDAs in the empty set model to a polynomial time solvable version by considering visibly DPDAs. 
\begin{theorem}
	\label{thm:Empty-P}
	The problems \textsc{Sync-DVPDA-Empty}, \textsc{Sync-DVCA-Empty}, and \textsc{Sync-DVVPDA-Empty} are decidable in polynomial time.
\end{theorem}
\begin{proof}
	We prove the claim for \textsc{Sync-DVPDA-Empty} as the other automata classes are sub-classes of DVPDAs.
	Let $M = (Q, \Sigma_{\text{call}} \cup \Sigma_{\text{int}} \cup \Sigma_{\text{ret}}, \Gamma, \delta, \bot)$ be a DVPDA. First, observe that if $\Sigma_{\text{ret}}$ is empty, then any synchronizing word $w$ for $M$ in the empty stack model cannot contain any letter from $\Sigma_{\text{call}}$. Hence, $M$ is basically a DFA and for DFAs the synchronization problem is in $\PTIME$~\cite{cerny1964,Vol2008,San2005}.
	From now on, assume $\Sigma_{\text{ret}} \neq \emptyset$. We show that a pair argument similar to the one for DFAs can be applied, namely that $M$ is synchronizable in the empty stack model if and only if every pair of states $p, q \in Q$ can be synchronized in the empty stack model. The only if direction is clear as every synchronizing word for $Q$ also synchronizes each pair of states. For the other direction, observe that since $M$ is a DVPDA, the stack height of each path starting in any state of $M$ is predefined by the sequence of input symbols. Hence, if we focus on the two runs starting in $p, q$ and ensure that their stacks are empty after reading a word $w$, then also the stacks of all other runs starting in other states in parallel are empty after reading $w$. Therefore, we can successively concatenate words that synchronize some pair of active states in the empty stack model and end up with a word that synchronizes all states of $M$ in the empty stack model. Further formal algorithmic details can be found in the appendix.
\begin{toappendix}
\begin{proof}[Formal and algorithmic proof details 	of Theorem~\ref{thm:Empty-P}]
In order to determine if a pair of states $p, q \in Q$ can be synchronized in the empty stack model, we build the following product automaton $M\times M[p, q] = (Q \times Q \cup Q, \Sigma_{\text{call}} \cup \Sigma_{\text{int}} \cup \Sigma_{\text{ret}}, \Gamma, \delta^2, (p, q), \bot, Q)$. For all states in $(r, s) \in Q \times Q$ for which $r \neq s$, $\delta^2$ simulates the actions of $\delta$ on $r$ in the first component and the actions of $\delta$ on $s$ in the second component. For states $(r, r) \in Q\times Q$, this is also the case for all transitions except for zero-tests of the stack, as here we transition to the corresponding state $r \in Q$.
	For $Q$, $\delta^2$, restricted to $Q$, is the same as $\delta$. 
	Clearly, $M\times M[p, q]$ accepts all words that have $w\sigma_r$ as a prefix, for which $w$ synchronizes the states $p$ and $q$ in $M$ in the empty stack model and $\sigma_r$ is any return letter in $\Sigma_{\text{ret}}$ that checks the empty stack condition. Further for all pairs of states $p, q$, $M \times M[p, q]$ is a DVPDA.
		As the emptiness problem for DVPDAs is in $\PTIME$~\cite{DBLP:conf/stoc/AlurM04}, we can build and test all product automata  $M \times M[p, q]$ for non-emptiness in polynomial time. 
\end{proof}
\end{toappendix}
\end{proof}
%
Does this mean everything is easy and we are done? Interestingly, the picture is not that simple, as considering the same and arbitrary stack models shows.
\begin{theorem}\label{Sync-DVPDA-Same-PSPACE}
	The problem \textsc{Sync-DVPDA-Same} is \PSPACE-hard.
\end{theorem}
\begin{proof}
	We reduce from 
\textsc{DFA-Sync-Into-Subset}. Let $A = (Q, \Sigma, \delta)$ be a DFA and $S \subseteq Q$. We construct from $A$ a DVPDA $M = (Q \cup \{q_S\}, \Sigma_{\text{call}} \cup \Sigma_{\text{int}} \cup \Sigma_{\text{ret}}, \{\smiley{}, \frownie{},\bot\}, \delta' = \delta'_{\text{c}} \cup \delta'_\text{i} \cup \delta'_\text{r}, \bot)$ with $q_S \notin Q$, $\Sigma_{\text{call}} = \{a\}$, $\Sigma_{\text{int}} = \Sigma$, $\Sigma_{\text{ret}} = \emptyset$ and $\Sigma_{\text{call}} \cap \Sigma_{\text{int}} = \emptyset$.
	The transition function $\delta'_\text{i}$ agrees with $\delta$ on all letters in $\Sigma_{\text{int}}$. For $q_S$ we set $\delta'_\text{c}(q_S, a) = (q_S, \smiley{})$ and $\delta'_\text{i}(q_S, \sigma) = q_S$ for all $\sigma \in \Sigma_{\text{int}}$. For $q \in S$, we set $\delta'_\text{c}(q, a) = (q_S, \smiley{})$, and for $q \notin S$, $\delta'_\text{c}(q, a) = (q, \frownie{})$.
	
	Note that $q_S$ is a sink-state and can only be reached from states in $S$ with a transition by the call-letter $a$. For states not in~$S$, the input letter $a$ pushes an $\frownie{}$ on the stack which cannot be pushed to the stack by any letter on a path starting in $q_S$. Hence, in order to synchronize~$M$ in the same stack model, a letter $a$ might only and must be read in a configuration where only states in $S \cup \{q_S\}$ are active. Every word $w \in \Sigma_{\text{int}}^*$ that brings $M$ in such a configuration also synchronizes $Q$ in $A$ into the set $S$.
\end{proof}
From the proof of Theorem~\ref{Sync-DVPDA-Same-PSPACE}, we can conclude the next results by observing that a DVPDA without any return letter cannot make any turn.
\begin{corollary}
	\label{cor:DVPDA-Same-NoReturn-PSPACE}
	\textsc{Sync-DVPDA-Same-NoReturn} and \textsc{0-Turn-Sync-DVPDA-Same} are \PSPACE-hard.
\end{corollary}

\noindent 
In contrast with the two previous results, \textsc{Sync-DVPDA-Same} is solvable in polynomial time if we have the promise that $\Sigma_{\text{ret}}\neq \emptyset$.
\begin{theorem}
	\label{thm:DVPDA-Same-Return-P}
	\textsc{Sync-DVPDA-Same-Return} is in \PTIME.
\end{theorem}
%
\begin{proof}
	We prove the claim by straight reducing to \textsc{Sync-DVPDA-Empty} with the identity function.
	If a DVPDA $M$ with $\Sigma_{\text{ret}}\neq \emptyset$ can be synchronized in the same stack model with a synchronizing word $w$, then $w$ can be extended to $ww'$ where $w'\in \Sigma_{\text{ret}}^*$ empties the stack. As~$M$ is deterministic and complete, $w'$ is defined on all states. As after reading~$w$, the stack content on all paths is the same, reading $w'$ extends all paths with the same sequence of states.
	Conversely, a word $w$ synchronizing a DVPDA $M$ with $\Sigma_{\text{ret}} \neq \emptyset$ in the empty stack model also synchronizes $M$ in the same stack model. 
\end{proof}
\noindent The arbitrary stack model requires the most interesting construction in the following proof.
\begin{theorem}
	\label{thm:DVPA-Arb-PSPACE}
	\textsc{Sync-DVPDA-Arb} is \PSPACE-hard.
\end{theorem}
\begin{proof}
	We give a reduction from the \PSPACE-complete problem \textsc{DFA-Sync-From-Subset}. 
	Let $A= (Q, \Sigma, \delta)$ be a DFA with $S \subseteq Q$. We construct from $A$ a DVPDA $M = (Q, \Sigma_{\text{call}} \cup \Sigma_{\text{int}} \cup \Sigma_{\text{ret}}, Q \cup \{\bot\}, \delta'=\delta'_\text{c} \cup \delta'_\text{i}\cup \delta'_\text{r}, \bot)$ where all unions in the definition of $M$ are disjoint. 
	Let $\Sigma_{\text{call}}=\Sigma$, $\Sigma_{\text{int}} = \emptyset$, and $\Sigma_{\text{ret}} = \{r\}$ with $r \notin \Sigma$.
	
	For states $s \in S$ we set $\delta'_\text{r}(s, r, \bot) = s$ and for states $q \in Q\backslash S$ we set $\delta'_\text{r}(q, r, \bot) = t$ for some arbitrary but fixed $t \in S$.
	For states $p, q \in Q$ we set $\delta'_\text{r}(q, r, p) = p$.
	
	For each call letter $\sigma \in \Sigma_{\text{call}}$ we set for $q \in Q$, $\delta'_\text{c}(q, \sigma) = (\delta(q, \sigma), q)$.
	
	First, assume $w$ is a word that synchronizes the set $S$ in the DFA $A$. Then, it can easily be observed that $rw$ is a synchronizing word for $M$ in the arbitrary stack model.
	
	Now, assume $w$ is a synchronizing word for $M$ in the arbitrary stack model. If $w\in \Sigma_{\text{call}}^*$, then $w$ is also a synchronizing word for $A$ and especially synchronizes the set $S$ in $A$. (*)
	Next, assume $w$ contains some letters $r$. 
The action of $r$ is designed such that it maps $Q$ to the set $S$ if applied to an empty stack and otherwise gradually undoes the transitions performed by letters from $\Sigma_{\text{call}}$. This is possible as each letter $\sigma \in \Sigma_{\text{call}}$ stores its pre-image on the stack when $\sigma$ is applied. Further, $r$ acts as the identity on the states in $S$ if applied to an empty stack. Hence, whenever the stacks are empty while reading some word, all states in $S$ are active. 

Hence, if $\sigma r$ is a subword of a synchronizing word $w=u\sigma rv$ of $M$, with $\sigma \in \Sigma_{\text{call}}$, then $w'=uv$ is also a synchronizing word of $M$. This justifies the set of rewriting rules $R=\{\sigma r\to\varepsilon\mid \sigma \in \Sigma_{\text{call}}\}$.
Now, consider a synchronizing word $w$  of $M$ where none of the rewriting rules from $R$ applies. Hence, $w\in\{r\}^*\Sigma_{\text{call}}^*$.
By (*), 
$w=r^kv$ with $k>0$ and $v\in\Sigma_{\text{call}}^*$. Then, $w'=rv$ is also a synchronizing word of $M$, because for all states $q \in Q$, $M$ is in the same configuration after reading $r$, starting in configuration $(q,\bot)$, as after reading $rr$. But as only (and all) states from $S$ are active after reading $r$, $v$ is also a word in $\Sigma^*$ that synchronizes the set $S$ in $A$.
\end{proof}
Observe that in the construction above, $\Sigma_{\text{ret}}\neq \emptyset$ for all input DFAs.
 The next corollary follows from Theorem~\ref{thm:DVPA-Arb-PSPACE} and should be observed together with the next theorem in contrast to Theorem~\ref{thm:DVPDA-Same-Return-P} and Corollary~\ref{cor:DVPDA-Same-NoReturn-PSPACE}.
\begin{corollary}
	\textsc{Sync-DVPDA-Arb-Return} is \PSPACE-hard.
\end{corollary}
\begin{theorem}
	\textsc{Sync-DVPDA-Arb-NoReturn} $\equiv$ \textsc{DFA-Sync}.
\end{theorem}
\begin{proof}
	Let $M$ be a DVPDA with empty set of return symbols.
	As there is no return-symbol, the transitions of $M$ cannot depend on the stack content. Hence, we can redistribute the symbols in $\Sigma_{\text{call}}$ into $\Sigma_{\text{int}}$ and obtain a DFA. The converse is trivial.
\end{proof}
If we move from deterministic visibly push-down automata to  even more restricted classes, like deterministic very visibly push-down automata or deterministic visibly counter automata, the three stack models do no longer yield synchronization problems with different complexities. Instead, all three models are equivalent, as stated next. Hence, their synchronization problems can be solved by the pair-argument presented in Theorem~\ref{thm:Empty-P} in polynomial time.
\begin{theorem}\label{Sync-equivalences}
	\textsc{Sync-DVCA-Empty} $\equiv$ \textsc{Sync-DVCA-Same} $\equiv$ \textsc{Sync-DVCA-Arb}.\\
	\textsc{Sync-DVVPDA-Empty} $\equiv$ \textsc{Sync-DVVPDA-Same} $\equiv$ \textsc{Sync-DVVPDA-Arb}.
\end{theorem}
\begin{proof}
	First, note that every DVCA is also a DVVPDA.
	If for a DVVPDA $\Sigma_{\text{ret}} \neq \emptyset$, then we can empty the stack after synchronizing the state set, as the very visibly conditions ensures that the contents of the stacks on all runs coincide. As the automaton is deterministic, all transitions for letters in $\Sigma_{\text{ret}}$ are defined on each state. As the stack content on all runs coincides in every step, the arbitrary stack model is identical to the same stack model and hence equivalent to the empty stack model.
	If $\Sigma_{\text{ret}} = \emptyset$, then we can reassign $\Sigma_{\text{call}}$ to $\Sigma_{\text{int}}$ in order to reduce from the same-stack and arbitrary stack to the empty stack variant, as transitions cannot depend on the stack content which is again the same on all runs due to the very visibly condition.
\end{proof}
\section{Restricting the Number of Turns Makes Synchronization Harder}
\label{sec:FiniteTurn}
We are now restricting the number of turns a synchronizing word may cause on any run. Despite the fact that we are hereby restricting the considered model even further, the synchronization problem becomes even harder, 
in contrast to the previous section.
\begin{theorem}
	\label{thm:n-turn-DVCA-hard}
	For every fixed $n\in \mathbb{N}$ with $n > 0$, the problems \textsc{$n$-Turn-Sync-DVCA-Same} and \textsc{$n$-Turn-Sync-DVCA-Arb} are \PSPACE-hard.
\end{theorem}
\begin{proof}
	We give a reduction from the \PSPACE-complete problem \textsc{DFA-Sync-Into-Subset}.
	Let $A= (Q, \Sigma, \delta)$ be a DFA with $S\subseteq Q$. We construct from $A$ a DVCA $M = (Q \cup \{q_{\text{sync}}\} \cup \{q_{\text{stall}_i} \mid 0 \leq i \leq n\}, \Sigma_{\text{call}} \cup \Sigma_{\text{int}} \cup \Sigma_{\text{ret}}, \{1, \bot\}, \delta'=\delta'_{\text{c}} \cup \delta'_{\text{i}} \cup \delta'_{\text{r}}, \bot\})$, where all unions are disjoint. We set $\Sigma_{\text{int}} = \Sigma$, $\Sigma_{\text{call}} = \{a\}$ and $\Sigma_{\text{ret}} = \{b\}$. For all internal letters, $\delta'_\text{i}$ agrees with $\delta$ on all states in~$Q$. For the letter $a$, we set for all $q\in S$, $\delta'_\text{c}(q, a) = (q_{\text{stall}_0}, 1)$ and for all $q \in Q \backslash S$, we set $\delta'_\text{c}(q, a) = (q, 1)$.
	For $b$ we loop in every state in $Q$.
	For $q_{\text{sync}}$, we loop with every letter in $q_{\text{sync}}$ (incrementing the counter with $a$ and decrementing it with $b$).

	Let $r$ be an arbitrary but fixed state in $Q$.
	For the states $q_{\text{stall}_i}$ we set for $i < n$, $\delta'_\text{c}(q_{\text{stall}_i}, a) = (q_{\text{stall}_i}, 1)$. Further, for even index $i<n$, we set $\delta'_\text{r}(q_{\text{stall}_i}, b, 1) = q_{\text{stall}_{i+1}}$ and $\delta'_\text{r}(q_{\text{stall}_i}, b, \bot) = r$.
	For odd index $i<n$, we set $\delta'_\text{r}(q_{\text{stall}_i}, b, 1) = r$, and $\delta'_\text{r}(q_{\text{stall}_i}, b, \bot) = q_{\text{stall}_{i+1}}$.
	For  even~$n$, let $\delta'_\text{c}(q_{\text{stall}_n}, a) = (q_{\text{sync}}, 1)$, $\delta'_\text{r}(q_{\text{stall}_n}, b, 1) = r$, and $\delta'_\text{r}(q_{\text{stall}_n}, b, \bot) = r$.
	For odd~$n$, let $\delta'_\text{c}(q_{\text{stall}_n}, a) = (q_{\text{stall}_n}, 1)$, $\delta'_\text{r}(q_{\text{stall}_n}, b, 1) = r$, and $\delta'_\text{r}(q_{\text{stall}_n}, b, \bot) = q_{\text{sync}}$.
%
	All other transitions (on internal letters) act as the identity.

	Observe that the state $q_{\text{sync}}$ must be the synchronizing state of $M$, since it is a sink state. 
	In order to reach $q_{\text{sync}}$ from any state in $Q$, the automaton must pass through all the states $q_{\text{stall}_i}$ for all $0 \leq i \leq n$ by construction. Since we can only transition from a state $q_{\text{stall}_i}$ to $q_{\text{stall}_{i+1}}$ with an empty or non-empty stack in alternation, passing the $q_{\text{stall}_i}$ gadget forces $M$ to make $n$ turns. For even $n$, the last upstroke is enforced by passing from $q_{\text{stall}_n}$ to $q_{\text{sync}}$ by explicitly increasing the stack. 
	As $M$ is only allowed to make $n$ turns while reading the $n$-turn synchronizing word this implies that any of the states $q_{\text{stall}_i}$ might be visited at most once, as branching back into $Q$ by taking a transition that maps to $r$ would force $M$ to go through all states $q_{\text{stall}_i}$ again, which exceeds the allowed number of strokes. Note that only counter values of at most one are allowed in any run which is currently in a state in $q_{\text{stall}_i}$ as otherwise the run will necessarily branch back into $Q$ later on.\footnote{In some states, such as $q_{\text{stall}_n}$ for even $n$, it is simply impossible to have a higher counter value.} Especially, this is the case for $q_{\text{stall}_0}$ which ensures that each $n$-turn synchronizing word has first synchronized $Q$ into $S$ before the first letter $a$ is read, as otherwise $q_{\text{stall}_0}$ is reached with a counter value greater than 1, or $M$ has already made a turn in $Q$ and hence cannot reach $q_{\text{sync}}$ anymore.
	
	In the construction above, for odd $n$ each run enters the synchronizing state with an empty stack (*). 
	For even $n$ each run enters the synchronizing state with a counter value of 1.
	The visibly condition, or more precisely very visibly condition as we are considering DVCAs, tells us that at each time while reading a synchronizing word, the stack content of every run is identical. In particular, this is the case at the point when the last state enters the synchronizing state and hence, any $n$-turn synchronizing word for $M$ is a synchronizing word in both the arbitrary and the same stack models.
%
%
\end{proof}
By observing that in the empty stack model allowing $n$ even turns is as good as allowing $(n-1)$ turns, essentially (*) from the previous proof yields the next result.
\begin{corollary}
	\label{cor:n-turn-empty-hard}
	For every fixed $n\in \mathbb{N}$ with $n > 0$, the problem \textsc{$n$-Turn-Sync-DVCA-Empty} is \PSPACE-hard.
\end{corollary}
\begin{toappendix}
\begin{proof}[Proof of Cor.~\ref{cor:n-turn-empty-hard}]
	Since we need to synchronize with an empty stack, for even $n$, the last upstroke cannot be performed. Hence, for even $n$, every DVCA $M$ can be synchronized by an $n$-turn synchronizing word if and only if $M$ can be synchronized by an $(n-1)$-turn synchronizing word. As for odd $n$ in the construction above, every $n$-turn synchronizing word synchronized~$M$ in the empty stack model, the claim follows from the proof in Theorem~\ref{thm:n-turn-DVCA-hard}.
\end{proof}
\end{toappendix}
\begin{corollary}
	\label{cor:N-Turn-very-visibly-hard}
	For every fixed $n\in \mathbb{N}$ with $n > 0$, the problems \textsc{$n$-Turn-Sync-DVPDA} and \textsc{$n$-Turn-Sync-DVVPDA} in the empty, same, and arbitrary stack models are  \PSPACE-hard.
\end{corollary}
\begin{theorem}
	\label{thm:0-turn-empty-P}
	\textsc{0-Turn-Sync-DVPDA-Empty} $\equiv$ \textsc{DFA-Sync}.
\end{theorem}
\begin{proof}
	The visibly condition and 
the fact that we can only synchronize with an empty stack means that we cannot read any letter from $\Sigma_{\text{call}}$, hence we cannot use the stack at all. Delete (a) all transitions with a symbol from $\Sigma_{\text{call}}$ and (b) all transitions with a symbol from~$\Sigma_{\text{ret}}$ and a non-empty stack. Then, assigning the elements in $\Sigma_{\text{ret}}$ to $\Sigma_{\text{int}}$ gives us a DFA.
\end{proof}
The next result is obtained by a reduction  from \textsc{DFA-Sync-From-Subset}.
\begin{theorem}
	\label{thm:0-turn-visibly-same}
	The problems \textsc{0-Turn-Sync-DVCA-\{Same, Arb\}}
	are \PSPACE-hard. 
\end{theorem}
\begin{toappendix}
\begin{proof}[Proof of Theorem~\ref{thm:0-turn-visibly-same}]
	We give a reduction from  the \PSPACE-complete problem \textsc{DFA-Sync-From-Subset}. Let $A = (Q, \Sigma, \delta)$ be a DFA with $S \subseteq Q$. We construct from $A$ a DVCA $M = (Q, \Sigma_{\text{call}} \cup \Sigma_{\text{int}} \cup \Sigma_{\text{ret}}, \{1, \bot\}, \delta'=\delta'_{\text{c}} \cup \delta'_{\text{i}} \cup \delta'_{\text{r}}, \bot)$. We set $\Sigma_{\text{call}} = \Sigma$, $\Sigma_{\text{int}} = \emptyset$, and $\Sigma_{\text{ret}} = \{b\}$. For all $q \in Q$ and $\sigma \in \Sigma_{\text{call}}$, we set $\delta'_\text{c}(q, \sigma) = (\delta(q, \sigma), 1)$. For all states $q \in Q\backslash S$, we set $\delta'_\text{r}(q, b, \bot) = s$ for some arbitrary but fixed state $s \in S$. All other transitions act as the identity.
	
	Note that the 0-turn condition only allows us to read the letter $b$ before any letter in $\Sigma_{\text{call}}$ has been read, as afterwards $b$ would decrease the stack after it has been increased. Therefore, every synchronizing word for $M$ in the same and arbitrary stack models also synchronizes~$S$ in~$Q$ by either synchronizing the whole set~$Q$ without using any $b$ transition, or it brings~$Q$ in exactly the set~$S$ with the first letter~$b$ and continues to synchronize~$S$.
\end{proof}
\end{toappendix}
\begin{corollary}
	The problems \textsc{0-Turn-Sync-DVVPDA-\{Same, Arb\}}, 
	and \textsc{0-Turn-Sync-DVPDA-\{Same, Arb\}} are \PSPACE-hard.
\end{corollary}
\section{(Non-)Tight Upper Bounds}
\label{sec:upperBound}
In this section we will prove that at least all considered problems are decidable (in contrast to non-visibly DPDAs and DCAs, see~\cite{dpda-crossref}) by giving exponential time upper bounds. We will also give some tight \PSPACE\ upper bounds for some \PSPACE-hard problems discussed in previous section, but for other problems previously discussed a gap between \PSPACE-hardness and membership in \EXP\ remains.
\begin{theorem}
	\label{thm:visibly-Exptime}
	All problems listed in Table~\ref{tab:VPA-results} are in \EXP.
\end{theorem}
\begin{proof}
	We show the claim explicitly for \textsc{Sync-DVPDA-Same}, \textsc{Sync-DVPDA-Arb}, \textsc{$n$-Turn-Sync-DVPDA-Empty}, \textsc{$n$-Turn-Sync-DVPDA-Same}, and \textsc{$n$-Turn-Sync-DVPDA-Arb}. The other results follow by inclusion of automata classes.
	
	Let $M = (Q, \Sigma_{\text{call}} \cup \Sigma_{\text{int}} \cup \Sigma_{\text{ret}}, \Gamma, \delta, \bot)$ be a DVPDA.
	We construct from $M$ the $|Q|$-fold product DVPDA $M^{|Q|}$ with state set $Q^{|Q|}$, consisting of $|Q|$-tuples of states, and alphabet $\Sigma_{\text{call}} \cup \Sigma_{\text{int}} \cup \Sigma_{\text{ret}}$. Since $M$ is a DVPDA, for every word $w\in (\Sigma_{\text{call}} \cup \Sigma_{\text{int}} \cup \Sigma_{\text{ret}})^*$, the stack heights on runs starting in different states in $Q$ is equal at every position in $w$. Hence, we can  multiply the stacks to obtain the stack alphabet $\Gamma^{|Q|}$ for $M^{|Q|}$. For the transition function $\delta^{|Q|}$ (split up into $\delta^{|Q|}_{\text{c}} \cup \delta^{|Q|}_{\text{i}}\cup \delta^{|Q|}_{\text{r}}$) of $M^{|Q|}$ we simulate $\delta$ independently on every state in an $|Q|$-tuple, i.e., for $(q_1, q_2, \dots, q_n) \in Q^{|Q|}$ and letters $\sigma_c \in \Sigma_{\text{call}}, \sigma_i \in \Sigma_{\text{int}}, \sigma_r \in \Sigma_{\text{ret}}$, we set
	\begin{itemize}
	\item 
	$\delta^{|Q|}_{\text{c}}((q_1, q_2, \dots, q_n), \sigma_c) = ((q_1', q_2', \dots, q_n'), (\gamma_1, \gamma_2, \dots, \gamma_n))$ if $\delta(q_j, \sigma_c) = (q_j', \gamma_j)$ for $j\in$$[n]$;
	\item
	$\delta^{|Q|}_{\text{i}}((q_1, q_2, \dots, q_n), \sigma_i) = (\delta(q_1, \sigma_i), \delta(q_2, \sigma_i), \dots, \delta(q_n, \sigma_i))$;
	\item
	$\delta^{|Q|}_{\text{r}}((q_1, q_2, \dots, q_n), \sigma_r, (\gamma_1, \gamma_2, \dots, \gamma_n)) = (\delta(q_1, \sigma_r, \gamma_1), \delta(q_2, \sigma_r, \gamma_2), \dots, \delta(q_n, \sigma_r, \gamma_n))$.
	\end{itemize}
	The bottom symbol of the stack is the $|Q|$-tuple $(\bot, \bot, \dots, \bot)$.
	Let $p_1, p_2, \dots, p_n$ be an enumeration of the states in $Q$ and set $(p_1, p_2, \dots, p_n)$ as the start state of $M^{|Q|}$. 

For \textsc{Sync-DVPDA-Arb}, 
set $\{(q, q, \dots, q) \in Q^{|Q|}\mid q\in Q\}$ as the final states for~$M^{|Q|}$.
	Clearly, for \textsc{Sync-DVPDA-Arb}, $M^{|Q|}$ is a DVPDA and the words accepted by $M^{|Q|}$ are precisely the synchronizing words for $M$ in the arbitrary stack model. As the emptiness problem can be decided for visibly push-down automata in time polynomial in the size of the automaton~\cite{DBLP:conf/stoc/AlurM04}, the claim follows,  observing that  $M^{|Q|}$ is exponentially larger than $M$.
	
	For \textsc{Sync-DVPDA-Same}, we produce a DVPDA $M^{|Q|}_\text{same}$ by enhancing the automaton $M^{|Q|}$ with three additional states $q_{\text{check}}$, $q_{\text{fin}}$, and $q_{\text{fail}}$ and an additional new return letter $r$ and set $q_{\text{fin}}$ as the single accepting state of $M^{|Q|}_\text{same}$, while the start state coincides with the one of $M^{|Q|}$.
	For states $(q_1, q_2, \dots, q_n) \in Q^{|Q|}$ we set $\delta^{|Q|}_{\text{r}}((q_1, q_2, \dots, q_n), r, (\gamma_1, \gamma_2, \dots, \gamma_n)) = q_{\text{check}}$ if $q_i = q_j$ and $\gamma_i = \gamma_j$, $\gamma_i \neq \bot$ for all $i, j \in [n]$. We set $\delta^{|Q|}_{\text{r}}((q_1, q_2, \dots, q_n), r, (\bot, \bot, \dots, \bot)) = q_{\text{fin}}$ if $q_i = q_j$ for all $i, j \in [n]$.
	For all other cases, we map with $r$ to $q_{\text{fail}}$. We let the transitions for $q_{\text{fail}}$ be defined such that $q_{\text{fail}}$ is a non-accepting trap state for all alphabet symbols. 
	For $q_{\text{check}}$ we set $\delta^{|Q|}_{\text{r}}(q_{\text{check}}, r, (\gamma_1, \gamma_2, \dots, \gamma_n)) = q_{\text{check}}$ if $\gamma_i = \gamma_j$ for $i,j \in [n]$. Further, we set $\delta^{|Q|}_{\text{r}}(q_{\text{check}}, r, (\bot, \bot, \dots, \bot)) = q_{\text{fin}}$ and map with $r$ to $q_{\text{fail}}$ in all other cases. The state $q_{\text{check}}$ also maps to $q_{\text{fail}}$ with all input symbols other than $r$.
	We let the transitions for $q_{\text{fin}}$ be defined such that $q_{\text{fin}}$ is an accepting trap state for all alphabet symbols.
	
	Clearly, for \textsc{Sync-DVPDA-Same} $M^{|Q|}_\text{same}$ is a DVPDA and the words accepted by $M^{|Q|}_\text{same}$ are precisely the synchronizing words for $M$ in the same stack model, potentially prolonged by a sequence of $r$'s, as the single accepting state $q_{\text{fin}}$ can only be reached from a state in $Q^{|Q|}$ where the states are synchronized and the stack content is identical for each run (which is checked in the state $q_{\text{check}}$). As the size of $M^{|Q|}_\text{same}$ is exponential in the size of $M$,
we get the claimed result as in the previous case.

For the \textsc{$n$-Turn} synchronization problems, we have to modify the previous construction by adding a stroke counter similar as in the proof of Theorem~\ref{thm:n-turn-DVCA-hard} (see Appendix~\ref{apx:n-turn-exptime}).
\begin{toappendix}
\begin{proof}[Further proof details for the \textsc{$n$-Turn} cases in Theorem~\ref{thm:visibly-Exptime}]
	\label{apx:n-turn-exptime}
	For the problems \textsc{$n$-Turn-\linebreak[4]Sync-DVPDA} in the empty, same, and arbitrary stack models, we enhance in $M^{|Q|}$ each $|Q|$-tuple with an additional index $I \in \{0,1,\dots, n+1\}$, i.e., the basic set of states is now $Q^{|Q|} \times \{0,1,\dots, n+1\}$. We further add for all three models the non-accepting trap state $q_{\text{fail}}$ to the set of states. For each $(|Q|+1)$-tuple, we implement the transition function $\delta^{|Q|}_{\text{i}}$ of $M^{|Q|}$ for internal letters in $\Sigma_{\text{int}}$ as before by keeping the value of the index $I$ in each transition.
	For call letters in $\Sigma_{\text{call}}$ we realize $\delta^{|Q|}_{\text{c}}$ as before for state-tuples with index $I < n+1$ by simulation $\delta$ on the individual states and setting in every image $I = I+1$ if $I$ is even, and keeping the value of $I$ if $I$ is odd. For tuples with index $I = n+1$, we proceed as before for smaller index if $n+1$ is odd, while for even $n+1$ we map with a call letter to the state $q_{\text{fail}}$. 
	For the return letters in $\Sigma_{\text{ret}}$, we realize $\delta^{|Q|}_{\text{r}}$ for pairs of states in $Q^{|Q|}\times \{1, 2, \dots, n+1\}$ and bottom of stack symbol $(\bot, \bot, \dots, \bot)$ as before by simulating $\delta$ on the individual states and keeping the value of $I$. 
	For all other stack symbols, we realize $\delta^{|Q|}_{\text{r}}$ as before for state-tuples with index $0 < I < n+1$ by simulation $\delta$ on the individual states and keeping the value of $I$ if	$I$ is even, and setting in every image $I = I+1$ if  $I$ is odd.
	For tuples with $I=n+1$ we proceed as before if $n+1$ is even.
	For states with index $I = 0$ or $I = n+1$ for odd $n+1$, we map with each return letter to $q_{\text{fail}}$ for stack symbols other than the bottom of stack symbol. 
	In all three models we set $(p_1, p_2, \dots, p_n, 0)$ as the start state, with $p_1, p_2, \dots, p_n$ being an enumeration of the states in~$Q$.
	
	For \textsc{$n$-Turn-Sync-DVPDA-Arb}, we set $\{(q, q, \dots, q, I) \mid q\in Q, I\in \{0,1,\dots,n+1\}\}\subset Q^{|Q|}\times \{0,1,\dots,n+1\}$ as the set of final states.
	
	For \textsc{$n$-Turn-Sync-DVPDA-Empty}, we set the additional trap state $q_{\text{fin}}$ as the single accepting state and add a new return letter $r$ with which we map to $q_{\text{fin}}$ for states $(q, q, \dots, q, I) \in Q^{|Q|}\times \{0,1,\dots,n+1\}$ with the bottom-of-stack symbol and to $q_{\text{fail}}$ for all other stack symbols or states.
	
	For \textsc{$n$-Turn-Sync-DVPDA-Same}, we add the two states $q_{\text{check}}$ and $q_{\text{fin}}$ to $M^{|Q|}$ and set $q_{\text{fin}}$ as the single accepting state. Again, we add a new return letter $r$.
	For states $(q, q, \dots, q, I)$ with $I \in \{0,1,\dots,n+1\}$ and symbol $(\gamma_1, \gamma_2, \dots, \gamma_n)$ on top of the stack, which is not the bottom-of-stack symbol, we map with $r$ to $q_{\text{check}}$ if all entries $\gamma_i$ in the stack symbol tuple are identical. If instead the bottom-of-stack symbol is on top of the stack, we map with $r$ for states $(q, q, \dots, q, I)$ directly to $q_{\text{fin}}$. For all other states and stack symbols, $r$ maps to $q_{\text{fail}}$.
	For $q_{\text{check}}$, we stay in $q_{\text{check}}$ with the letter $r$ if we see a symbol $(\gamma, \gamma, \dots, \gamma)$ on the stack with $\gamma \in \Gamma\backslash\{\bot\}$ and map with $r$ to $q_{\text{fin}}$ if we see the bottom-of-stack symbol. For all other stack symbols, $r$ maps $q_{\text{check}}$ to $q_{\text{fail}}$. Also, all input letter other than $r$ maps $q_{\text{check}}$ to $q_{\text{fail}}$. For $q_{\text{fin}}$ we define all transitions such that $q_{\text{fin}}$ is a trap state.
	
	In all three cases, the constructed automaton 
is a DVPDA that accepts precisely the $n$-turn synchronizing words for $M$ (potentially prolonged by a sequence of $r$'s) in the respective stack model. As the  constructed automaton 
is of size $\mathcal{O}((|Q||\Gamma|)^{|Q|})$ in all three cases, we can decide whether the constructed automaton 
 accepts at least one word in time exponential in the description of $M$.
\end{proof}
\end{toappendix}
\end{proof}

\begin{remark}
	It cannot be expected to show \PSPACE-membership of synchronization problems concerning DVPDAs using a $|Q|$-fold product DVPDA, as the resulting automata is exponentially large in the size of the DVPDA that is to be synchronized, as the emptiness problem for DVPDAs is \PTIME-complete~\cite{DBLP:journals/sigact/OkhotinS14}. Rather, one would need a separate membership proof. 
	We conjecture that a \PSPACE-membership proof similar to the one for the short synchronizing word problem presented in~\cite{DBLP:journals/jalc/ChistikovMS19} can be obtained if exponential upper bounds for the length of synchronizing words for DVPDAs in the respective models can be obtained.
\end{remark}
\begin{theorem}
	\label{thm:DVPDA-0-same-PSPACE}
	The problems
	\textsc{0-Turn-Sync-\{DVPDA, DVVPDA, DVCA\}-Same}
	are in \PSPACE.
\end{theorem}
\begin{proofsketch}
	Let $M = (Q, \Sigma_{\text{call}} \cup \Sigma_{\text{int}} \cup \Sigma_{\text{ret}}, \Gamma, \delta, \bot)$ be a DVPDA.
	For the same stack model, the 0-turn condition forbids us to put in simultaneous runs different letters on the stack at any time while reading a synchronizing word, as we cannot exchange symbols on the stack with visible PDAs. Note that this is a dynamic runtime-behavior and does not imply that $M$ is necessarily very visibly. 
	Further, the 0-turn and visibility condition enforces that at each step the next transition does not depend on the stack content if the symbol on top of the stack is not~$\bot$.
	Hence, we can construct from $M$ a $|Q|$-fold DFA (with a state set exponential in the size of $|Q|$) in a similar way as in the proof of Theorem~\ref{thm:visibly-Exptime} by neglecting the stack as nothing is ever popped from the stack. Details on the construction can be found in the appendix. As the emptiness problem for DFAs can be solved in \NLOGSPACE, the claim follows with Savitch's famous theorem stating that $\NPSPACE = \PSPACE$~\cite{DBLP:journals/jcss/Savitch70}.\footnote{Here, a smaller powerset-construction would also work but for simplicity, we stuck with the introduced $|Q|$-fold product construction.}
\end{proofsketch}
\begin{toappendix}
\begin{proof}[Proof of Theorem~\ref{thm:DVPDA-0-same-PSPACE}]
	Let $M = (Q, \Sigma_{\text{call}} \cup \Sigma_{\text{int}} \cup \Sigma_{\text{ret}}, \Gamma, \delta=\delta_{\text{c}} \cup \delta_{\text{i}} \cup \delta_{\text{r}}, \bot)$ be a DVPDA.
	For the same stack model, the 0-turn condition forbids us to put in simultaneous runs different letters on the stack at any time while reading a synchronizing word, as we cannot exchange symbols on the stack with visible PDAs. Note that this is a dynamic runtime-behavior and does not imply that $M$ is necessarily very visibly. 
	Further, the 0-turn and visibility condition enforces that at each step the next transition does not depend on the stack content if the symbol on top of the stack is not~$\bot$.
	We construct from $M$ a partial $|Q|$-fold product DFA $M^{|Q|}$ with state set $Q^{|Q|}\times \{0,1\}$, consisting of $|Q|$-tuples of states with an additional bit of information which will indicate whether the stack is still empty, and alphabet $\Sigma_{\text{call}} \cup \Sigma_{\text{int}} \cup \Sigma_{\text{ret}}$. For the transition function $\delta^{|Q|}$ of $M^{|Q|}$, we simulate for a state $(q_1, q_2, \dots, q_n, b)$ with $q_1, q_2, \dots, q_n \in Q$, $b \in \{0,1\}$ and letter $\sigma$, $\delta$ (by restricting the image to the first component in $Q$ for call letters) on the individual states $q_i$, $q_j$ with $i,j \in [n]$ in the tuple if (1) $\sigma \in \Sigma_{\text{ret}}$ and $b=0$,
	(2) $\sigma \in \Sigma_{\text{int}}$, or
	(3) $\sigma \in \Sigma_{\text{call}}$ and for $\delta_\text{c}(q_i, \sigma) = (q_i', \gamma_i)$, $\delta_\text{c}(q_j, \sigma) = (q_j', \gamma_j)$ it holds that $\gamma_i = \gamma_j$. In case (1) and (2), we keep the value of~$b$ in the transition and in case (3), we ensure $b=1$ in the image of the transition.
	The size of the state graph of $M^{|Q|}$ is bounded by $\mathcal{O}(|Q|^{|Q|})$. Clearly, the DVPDA $M$ can be synchronized by a 0-turn synchronizing word in the same stack model if and only if there is a path in the state graph of $M^{|Q|}$ from the state $(q_1, q_2, \dots, q_n, 0)$ for $Q = \{q_1, q_2, \dots, q_n\}$ to some state in $\{(q_i, q_i, \dots, q_i, b) \mid i\in [n], b \in \{0,1\}\}$. These $2|Q|$ reachability tests can be performed in $\NPSPACE = \PSPACE$~\cite{DBLP:journals/jcss/Savitch70}. 
	The claim for the other problems follows by inclusion of automata classes.
\end{proof}
\end{toappendix}
\begin{corollary}\label{DVPDA-Same-NoReturn-PSPACE}
	\textsc{Sync-DVPDA-Same-NoReturn} is in \PSPACE.
\end{corollary}
\begin{toappendix}
\begin{proof}[Proof of Cor.~\ref{DVPDA-Same-NoReturn-PSPACE}]
	Let $M = (Q, \Sigma_{\text{call}} \cup \Sigma_{\text{int}} \cup \Sigma_{\text{ret}}, \Gamma, \delta, \bot)$ be a DVPDA with $\Sigma_{\text{ret}} = \emptyset$. As we have no return letter, any synchronizing word for $M$ is also a 0-turn synchronizing word and hence, the claim follows with Theorem~\ref{thm:DVPDA-0-same-PSPACE}.
\end{proof}
\end{toappendix}
\begin{theorem}
	The problems
	\textsc{0-Turn-Sync-\{DVPDA, DVVPDA, DVCA\}-Arb}, and
	\textsc{1-Turn-Sync-DVCA-\{Empty, Same, Arb\}}
	are in \PSPACE.
\end{theorem}
\begin{proof}
	The claim follows from \cite[Theorem 16 \& 17]{dpda-crossref} by inclusion of automata classes.
\end{proof}


\section{Sequential Transducers}
\label{sec:Transducer}
In~\cite{dpda-crossref}, the concept of trace-synchronizing a sequential transducer has been introduced. We want to extend this concept to sequential transducers showing some kind of \emph{visible behavior} regarding their output, inspired by the predetermined stack height behavior of DVPDAs. 
We call $T = (Q, \Sigma, \Gamma, q_0, \delta, F)$ a \emph{sequential transducer} (ST for short)
if $Q$ is a finite set of states, $\Sigma$ is an 
input alphabet, $\Gamma$ is an 
output alphabet, $q_0$ is the start state, $\delta \colon Q \times \Sigma \to Q \times \Gamma^*$ is a total transition function, and $F$ collects the final states. 
We generalize $\delta$ from input letters to words by concatenating the produced outputs. 
$T$ is called a \emph{visibly sequential transducer} (VST for short) [or \emph{very visibly sequential transducer} (VVST for short)] if for each $\sigma \in \Sigma$ and for all $q_1, q_2 \in Q$ and $\gamma_1,\gamma_2\in\Gamma^*$, it holds that $\delta(q_1, \sigma) = (q_1', \gamma_1)$ and $\delta(q_2, \sigma) = (q_2', \gamma_2)$ implies that $|\gamma_1| = |\gamma_2|$ [or that $\gamma_1 = \gamma_2$, respectively].
A VVST $T$ is thereby computing the same homomorphism $h_T$, regardless of which states are chosen as start and final states (*). Hence, if $A_T$ is the underlying DFA (ignoring any outputs), then $h_T(\lang(A_T))\subseteq\Gamma^*$ describes the language of all possible outputs of $T$. By Nivat's theorem~\cite{Niv68}, a language family is a full trio iff 
it is closed under VVST and inverse homomorphisms. Our considerations also show that  a language family is a full trio iffit is closed under VVST and inverse VVST mappings.

We say that a word $w$ \emph{trace-synchronizes} a sequential transducer $T$ if, for all states $p, q \in Q$, 
$\delta(p, w) = \delta(q, w)$,
i.e.,  a synchronizing state is reached, producing identical output.
Notice that from the viewpoint of trace-synchronization, we do not assume that a VVST has only one state. 

\begin{definition}[\sc Trace-Sync-Transducer]
	\ \\
	Given: Sequential transducer $T = (Q, \Sigma, \Gamma, \delta)$.\\
	Question: Does there exists a word $w \in \Sigma^*$ that trace-synchronizes $T$?
\end{definition}

\nocite{DBLP:books/lib/Berstel79,Ginsburg66,DBLP:journals/tcs/Carton07,sakarovitch2003elements}
\begin{toappendix}
\noindent{\sffamily\textcolor{darkgray}{\textbf{Remarks on sequential transducers.}}}
The definitions in the literature are not very clear for finite automata with outputs. We follow here the name used by Berstel in~\cite{DBLP:books/lib/Berstel79}; Ginsburg~\cite{Ginsburg66} called Berstel's sequential transducers \emph{generalized machines}, but used the term \emph{sequential transducer} for the nondeterministic counterpart.
For non-deterministic transducers which allow to read multiple letters at once the concept of fixing the ratio between the length of the produced output and the length of the input was already studied in~\cite{DBLP:journals/tcs/Carton07} and was even mentioned in~\cite{sakarovitch2003elements}. Here, the ratio is fixed for every transition independent of the input letter(s) and a transducer admitting such a fixed ratio $\alpha$ is called \emph{$\alpha$-synchronous}. The term 'synchronization' again appears here but refers to finding an $\alpha$-synchronous transducer to a given rational relation.
\end{toappendix}

We define 
\textsc{Trace-Sync-VST} and \textsc{Trace-Sync-VVST} by considering a VST, respectively VVST, instead.
In contrast to the undecidability of \textsc{Trace-Sync-Transducer}~\cite{dpda-crossref}, we get the following results for trace-synchronizing VST and VVST from previous results. 
\begin{theorem}\label{Sync-VST-PSPACE}
	\textsc{Trace-Sync-VST} is \PSPACE-complete.
\end{theorem}
\begin{proof}
	First, observe that there is a straight reduction from the problem \textsc{Sync-DVPDA-Same-NoReturn} to \textsc{Trace-Sync-VST} as the input DVPDAs to the problem \textsc{Sync-DVPDA-Same-NoReturn} have no return letters and hence, the stack is basically a write only tape. Further, as the remaining alphabet is partitioned into letters in $\Sigma_{\text{call}}$, which write precisely one symbol on the stack, and into letters in $\Sigma_{\text{int}}$, writing nothing on the stack, the visibly condition is satisfied when interpreting the DVPDA with $\Sigma_{\text{ret}}=\emptyset$ as a VST.
	
	There is also a straight reduction from \textsc{Trace-Sync-VST} to \textsc{Sync-DVPDA-Same-NoReturn} as follows. For a VST $T =(Q, \Sigma, \Gamma, \delta)$ we construct a DVPDA $M = (Q, \Sigma_{\text{call}} \cup \Sigma_{\text{int}}, \Gamma', \delta)$ with $\Sigma_{\text{ret}}=\emptyset$ by introducing for each $\sigma \in \Sigma$ a new alphabet $\Sigma_\sigma=\{w\in \Gamma^* \mid \exists q, q' \in Q \colon \delta(q, \sigma) = (q', w)\}$. Observe that $\Sigma_\sigma$ is either $\{\epsilon\}$ or contains only words of the same length. By setting $\Sigma_{\text{int}} = \{\sigma \in \Sigma\mid \Sigma_\sigma=\{\epsilon\}\}$, $\Sigma_{\text{call}} = \{\sigma \in \Sigma \mid \Sigma_\sigma \neq \{\epsilon\}\}$, $\Gamma' = \bigcup_{\sigma \in \Sigma} (\Sigma_\sigma\backslash \{\epsilon\})$, and interpreting the output sequence $w \in \Gamma^*$ produced by $\delta$ as the single stack symbol in $\Gamma'$.
\end{proof}

\noindent
Yet, 
by Observation (*), we inherit from \textsc{Sync-DFA} the following algorithmic result.
\begin{theorem}\label{Sync-VVST-PTIME}
	\textsc{Trace-Sync-VVST} is in \PTIME.
\end{theorem}

\begin{toappendix}
\begin{proof}[Proof of Theorem~\ref{Sync-VVST-PTIME}]
	For each VVST $T=(Q, \Sigma, \Gamma, \delta)$ and $w \in \Sigma^*$ the same output is already produced in $\delta(q, w)$ for all $q \in Q$. Hence, we can ignore the output and test $T$ for trace-synchronization by the polynomial time pair-algorithm for DFAs~\cite{San2005}.
\end{proof}
\end{toappendix}

\section{Discussion}
\label{sec:Discussion}
Our results 
concerning DVPDAs and sub-classes thereof, are subsumed in Table~\ref{tab:VPA-results}. While all problems listed in the table are contained in \EXP, the table lists several problems for which their known complexity status still contains a gap between \PSPACE-hardness lower bounds and  \EXP\ upper bounds. Presumably, their precise complexity status is closely related to upper bounds on the length of synchronizing words which we want to consider in the near future.
One of the questions which could be solved in this work is if there is a difference between the complexity of synchronization in the same stack model and synchronization in the arbitrary stack model. While for general DPDA, DCA, and sub-classes thereof, see~\cite{dpda-crossref}, these two models admitted synchronization problems with the same complexity, here we observed that these models  can differ significantly. 
While the focus of this work is on determining the complexity status of synchronizability
for different models of automata, an obvious question for future research is the complexity status of closely related, and well understood questions in the realm of DFAs, such as the problem of shortest synchronizing word,
  subset synchronization, synchronization into a subset, and careful synchronization.

Here is one subtlety that comes with shortest synchronizing words: While for finding synchronizing words of length at most $k$ for DFAs, it does not matter if the number $k$ is given in unary or in binary due to the known cubic upper bounds on the lengths of shortest synchronizing words, this will make a difference in other models where such polynomial length bounds are unknown. More precisely, for instance with DVPDAs, it is rather obvious that with a unary length bound~$k$, the problem becomes \NP-complete, while the status is unclear for binary length bounds. As there is no general polynomial upper bound on the length of shortest synchronizing words for VPDAs, they might be of exponential length. Hence, we don not get membership in \PSPACE\ easily, not even for synchronization models concerning DVPDA for which general synchronizability is solvable in \PTIME, as it might be necessary to store the whole word on the stack in order to test its synchronization effects. 
\newpage
\bibliography{bib}

\end{document}